\documentclass[english]{article}
\usepackage[T1]{fontenc}      
\usepackage[english]{babel}
\usepackage[utf8]{inputenc}
\usepackage[style=alphabetic,firstinits=true,maxnames=6,minnames=5, 
     maxbibnames=99,natbib=true, 
     url=false,isbn=false,doi=false, 
     backend=bibtex]{biblatex}
\usepackage[list=off,singlelinecheck=off]{caption}
\usepackage{microtype}
\usepackage[section]{placeins}

\def\cf{cf.\xspace}

\usepackage{amsthm,amsmath,amsfonts,amssymb,stmaryrd}   
\usepackage{textcomp}         
\usepackage{xcomment}
\usepackage{xspace}
\usepackage{url}
\usepackage{subfig}
\usepackage[colorlinks=true, 
            linkcolor=black, 
            citecolor=black, 
            urlcolor=black, 
            pdftitle={Turing Degrees of limit sets of cellular automata},
            pdfauthor={Alex Borello, Pascal Vanier and others}, 
            pdfkeywords={Cellular Automata, Turing degrees, computability, limit sets}
           ]{hyperref}

\usepackage{tikz,ifthen}
\usetikzlibrary{decorations.pathmorphing,decorations.pathreplacing,arrows,
automata,patterns,
positioning}

\date{}

\newcommand{\ZZ}{\ensuremath{\mathbb{Z}}}
\newcommand{\ZZd}{\ensuremath{{\mathbb{Z}^d}}\xspace}

\newcommand{\NN}{\ensuremath{\mathbb{N}}}

\newcommand{\cantor}{\ensuremath{\left\{0,1\right\}^\NN}\xspace}

\newcommand{\CBd}[1]{\ensuremath{\mathfrak D\left(#1\right)}}
\newcommand{\CBdn}[2]{\ensuremath{{#1}^{\left(#2\right)}}}
\newcommand{\CB}[1]{\ensuremath{\mathfrak{CB}\left(#1\right)}}

\newcommand{\pizu}{\ensuremath{\Pi^0_1}\xspace}

\newcommand{\turinf}{{\ensuremath{\leq_T}}}

\newcommand{\tursup}{{\ensuremath{\geq_T}}}
\newcommand{\turequiv}{{\ensuremath{\equiv_T}}}
\newcommand{\turdeg}{{\ensuremath{\deg_T}}}
\newcommand{\turdegzero}{{\ensuremath{\mathbf{0}}}}

\newcommand{\halts}[2][truie]{\ifthenelse{\equal{#1}{truie}}{#2\!\downarrow}{#2\!\downarrow^{#1}}}

\newcommand{\ca}[1]{\ensuremath{\mathcal{#1}}}
\newcommand{\limitset}[1]{\ensuremath{\Omega\left(#1\right)}}
\newcommand{\cacf}[1]{\ensuremath{\mathfrak{#1}}}

\newcommand{\inter}[2]{\llbracket #1, #2\rrbracket}

\usepackage[bold,basic]{complexity}

\newcommand{\ie}{i.e.\xspace}

\newcommand{\includepicture}[1]{\includegraphics{pix/#1}}

\newtheorem{maintheorem}{Theorem}
\newtheorem{theorem}{Theorem}[section]
\newtheorem{corollary}[theorem]{Corollary}
\newtheorem{lemma}[theorem]{Lemma}

\theoremstyle{remark}

\theoremstyle{definition}

\title{Turing degrees of limit sets of cellular automata}

\author{Alex Borello, Julien Cervelle, Pascal Vanier}

\begin{document}

\maketitle

\begin{abstract}
    Cellular automata are discrete dynamical systems and a model of computation. The limit 
    set of a cellular automaton consists of the configurations having an infinite sequence 
    of preimages. It is well known that these always contain a computable point and that
    any non-trivial property on them is undecidable. We go one step further in this article
    by giving a full characterization of the sets of Turing degrees of cellular automata:
    they are the same as the sets of Turing degrees of effectively closed sets containing a
    computable point.
\end{abstract}

\section{Introduction}
Cellular Automata (CAs for short) are both discrete dynamical systems and a model of computation. They
were introduced in the late 1940s independently by John von~Neumann and Stanislaw Ulam to
study, respectively, self-replicating systems and the growth of quasi-crystals.

A $d$-dimensional CA consists of cells aligned on~\ZZd that may be in a finite number of states, and  are updated synchronously with a local rule, \ie depending only on a finite neighborhood.
All cells operate under the same local rule. The state of all cells at some time step is called
a configuration. CAs are very well known for being simple systems that may exhibit complicated 
behavior. 

A $d$-dimensional subshift of finite type (SFT for short) is a set of colorings of \ZZd by a finite number of
colors containing no pattern from a finite family of forbidden patterns. Most proofs of undecidability
concerning CAs involve the use of SFTs, so both topics are very 
intertwined~\cite{Kar1990,Kar1992,Kar1994,Mey2008,Kar2011}. A recent trend in
the study of SFTs has been to give computational characterizations of dynamical properties, which
has been followed by the study of their computational structure and in particular the comparison
with the computational structure of effectively closed sets, which are the subsets of \cantor on
which some Turing machine does not halt. It is quite easy to see that SFTs are such sets.

In this paper, we follow this trend and study the limit set $\limitset{\ca A}$ of a CA $\ca A$, which consist of all the 
configurations of the CA that can occur after arbitrarily long computations. They were introduced
by \citet{CPY1989} in order to classify CAs. It has
been proved that non-trivial properties on these sets are undecidable 
by \citet{Kar1994b,GR2010} for CAs of all dimensions. Limit sets of CAs are subshifts, and
the question of which subshifts may be limit sets of CA has been a thriving topic, see
\cite{Hur1987,Hur1990b,Hur1990,Maa1995,FK2007,DiLM2009,BGK2011}. However, most of these results 
are on the language of the limit set or on simple limit sets. Our aim here is to study the
configurations themselves. 

In dimension~$1$, limit sets are effectively closed sets, so it is quite 
natural to compare them from a computational point of view. The natural measure of
complexity for effectively closed sets is the Medvedev degree \cite{Sim2011a}, which, 
informally, is a measure of the complexity of the simplest points of the set. As limit 
sets always contain a uniform configuration (wherein all cells are in the 
same state), they always contain a computable point and 
have Medvedev degree 
\turdegzero. Thus, if we want to study their computable structure, we need a finer measure; in
this sense, the set of Turing degrees is appropriate. 

It turns out that for SFTs, there is a
characterization of the sets of Turing degrees found by \citet{JeandelV2013:turdeg}, which states
that one may construct SFTs with the same Turing degrees as any effectively closed set 
containing a computable point. In the case of limit sets, such a characterization would
be perfect, as limit sets always contain a computable point\footnote{Note that this is not the case
for subshifts: there exist non-empty subshifts containing only non-computable points.}. This is
exactly what we achieve in this article:

\begin{maintheorem}\label{mainthm}
 For any effectively closed set $S$, there exists a cellular automaton~$\ca A$ such that
 \[
     \turdeg\limitset{\ca A}=\turdeg{S}\cup\{\turdegzero\}\text{.}
 \]
\end{maintheorem}

In the way to achieve this theorem, we introduce a new construction which gives us some 
control over the limit set. We hope that this construction will lead to other unrelated results
on limit sets of CAs, as it was the case for the construction in \cite{JeandelV2013:turdeg}, see
\cite{JeandelV2013}. 

The paper is organized as follows. In Section~\ref{prelim} we recall the usual definitions concerning
CAs and Turing degrees. In Section~\ref{requirements} we give the reasons for each trait of the
construction which allows us to prove theorem~\ref{mainthm}. In Section~\ref{construction} we 
give the actual
construction. We end the paper by a discussion, in Section~\ref{CB}, on the Cantor-Bendixson ranks of the 
limit sets of CAs. The choice has been made to have colored figures, which are best viewed on screen.

\section{\label{prelim}Preliminary definitions}

A ($1$-dimensional) \emph{cellular automaton} is a triple $\ca A = (Q, r, \delta)$, where $Q$ is the finite set of \emph{states}, 
$r > 0$ is the \emph{radius} and $\delta : Q^{2r + 1}\to Q$ the \emph{local transition function}.

An element of~$i\in\ZZ$ is called a \emph{cell}, and the set $\inter{i - r}{i + r}$ is the
\emph{neighborhood} of~$i$ (the elements of which are the \emph{neighbors} of~$i$). A
\emph{configuration} is a function $\cacf c : \ZZ\to Q$. The local transition function induces a
\emph{global transition function} (that can be regarded as the automaton itself, hence the notation),
which associates to any configuration~$\cacf c$ its \emph{successor}:
\[
	\ca A(\cacf c) : \left\{\begin{array}{ccl}
		\ZZ &\to& Q\\
		i &\mapsto& \delta(\cacf c(i - r), \dots, \cacf c(i - 1), \cacf c(i), \cacf c(i + 1), \dots, \cacf c(i + r)))\text{.}
	\end{array}\right.
\]
In other words, all cells are finite automata that update their states in parallel, according to the same local transition rule, 
transforming a configuration into its successor.

If we draw some configuration as a horizontal bi-infinite line of cells, then add its successor above it, then the successor of 
the latter and so on, we obtain a \emph{space-time diagram}, which is a two-dimensional representation of some 
computation performed by~$\ca A$.

A \emph{site} $(i, t)\in\ZZ^2$ is a cell~$i$ at a certain time step~$t$ of the computation we consider 
(hereinafter there will never be any ambiguity on the
automaton nor on the computation considered).

The \emph{limit set} of~$\ca A$, denoted by $\limitset{\ca A}$, is the set of all the configurations that can
appear after arbitrarily many computation steps:
\[
	\limitset{\ca A} = \bigcap_{k\in\NN}\ca A^k(Q^\ZZ)\text{.}
\]

For surjective CAs, the limit set is the set of all possible configurations $Q^\ZZ$, while
for non-surjective CAs, it is the set of all configurations containing no 
orphan of any order, see \cite{Hur1990b}. An \emph{orphan of order~$n$} is a finite word~$w$ 
which has no preimage by $\ca A^n_{|Q^{|w|}}$.

An \emph{effectively closed set}, or \emph{\pizu class}, is a subset~$S$ of \cantor for which there
exists a Turing machine that, given any $x\in\cantor$, halts if and only if $x\not\in S$.
Equivalently, a class $S\subseteq\cantor$ is \pizu if there exists a computable set~$L$
such that $x\in S$ if and only if no prefix of~$x$ is in~$L$. 
It is then quite easy to see that limit sets of CAs are \pizu classes: for any limit set, the set of forbidden patterns is the set of all orphans of all orders, which form a recursively enumerable set, since it is
computable to check whether a finite word is an orphan.

For $x, y\in\cantor$, we say that $x\turinf y$ if $x$ is computable by a Turing machine using $x$
as an oracle. If $x\turinf y$ and $x\tursup y$, $x$ and~$y$ are said to be Turing-equivalent, 
which is noted $x\turequiv y$. The \emph{Turing degree} of~$x$, noted $\turdeg x$, is its equivalence
class under relation~$\turequiv$. The Turing degrees form a lattice whose bottom is \turdegzero,
the Turing degree of computable sequences.

Effectively closed sets are quite well understood from a computational point of view, and 
there has been numerous contributions concerning their Turing degrees, see the book of
\citet{CR1998} for a survey. One of the most interesting results may be that there exist 
\pizu classes whose members are two-by-two Turing incomparable \cite{JS1972}.

\section{\label{requirements}Requirements of the construction}

The idea to prove Theorem~\ref{mainthm} is to make a construction that embeds computations of a Turing 
machine that will check a read-only oracle tape containing a member of the 
\pizu class $S$ that will have to appear ``non-deterministically''. The following constraints have to be addressed.
\begin{itemize}
    \item Since CAs are intrinsically deterministic, this non-determinism will
have to come from the ``past'', \ie 
        from the ``limit'' of the preimages. 
    \item The oracle tape, the element of \cantor that needs to be checked, 
        needs to appear entirely on at least one 
        configuration of the limit set.
    \item Each configuration of the limit set containing the oracle tape 
        needs to have exactly one head of the Turing machine, in order 
        to ensure that there really is a computation going on in the associated space-time diagram.
    \item The construction, without any computation, needs to have a very simple limit set, 
        \ie it needs to be computable, and in particular countable; this to ensure that no complexity 
        overhead will be added to any configuration containing the oracle tape, and that ``unuseful'' 
        configurations of the limit set --~the configurations that do not appear in a 
        space-time diagram corresponding to a computation~-- will be computable.
    \item The computation of the embedded Turing machine needs to go backwards, this to ensure  
        that we can have the non-determinism. And an error in the computation must ensure
        that there is no infinite sequence of preimages.
    \item The computation needs to have a beginning (also to ensure the presence of a head), 
        so the construction needs some marked beginning, and the representation of the oracle and work tapes in the construction need to disappear at this point, otherwise by compactness the part without any computation could 
        be extended bi-infinitely to contain any member of \cantor, thus leading to the full set
        of Turing degrees.
\end{itemize}
There are other constraints that we will discuss during the construction, as they arise.

In order to make a construction complying to all these constraints, we reuse, with heavy modifications,
an idea of~\citet{JeandelV2013:turdeg}, which is to construct a sparse grid. 
However, their construction, being meant for subshifts, requires to be completely 
rethought in order to work for CAs. In particular, there was no determinism in this construction, 
and the oracle tape did not need to appear on a single column/row, since their result was on two-dimensional subshifts.

\section{\label{construction}The construction}
\subsection{\label{sparsegrid}A self-vanishing sparse grid}

In order to have space-time diagrams that constitute sparse grids, the idea is to have columns 
of squares, each of these columns containing less and less squares as we move to the left, see
fig.~\ref{butterfly:baselayer}. The CA
has three categories of states:
\begin{itemize}
    \item a \emph{killer state}, which is a spreading state that erases anything on its path;
    \item a \emph{quiescent state}, represented in white in the figures; its sole purpose is to mark the spaces that are ``outside'' the construction;
    \item some \emph{construction states}, which will be constituted of signals and background colors.
\end{itemize}

In order to ensure that just with the signals themselves it is not possible to encode anything 
non-computable in the limit set, all signals will need to have, at all points, at any time, different 
colors on their left 
and right, otherwise the local rule will have a killer state arise. Here are the main signals.
\begin{itemize}
    \item Vertical lines: serve as boundaries between columns of squares and
        form the left/right sides of the squares.
    \item SW-NE and SE-NW diagonals: used to mark the corners of the squares, they are
        signals of respective speeds $1$ and $-1$. Each time
        they collide with a vertical line (except for the last square of the row), they bounce
        and start the converse diagonal of the next square.
    \item Counting signal: will count the number of squares inside a column; every time 
        it crosses the SW-NE diagonal of a square it will shift to the left. When it is superimposed 
        to a vertical line, it means that the square is the last of its column, so when it crosses
        the next SE-NW diagonal, it vanishes and with it the vertical line.
    \item Starting signals: used to start the next column to the left, at the bottom of one
        column. Here is how they work.
        \begin{itemize}
            \item The bottommost signal, of speed~$-\frac 14$, is at the boundary between the empty part of the
                space-time diagram and the construction. It is started $4$~time steps after the collision
                with the signal of speed~$-\frac 13$.
            \item The signal of speed~$-\frac 13$ is started just after the vertical line sees the incoming SE-NW diagonal of the first square 
                of the row on the right, at distance~$3$\footnote{That can be done, provided the radius of the CA is large enough.} (the diagonal will collide with the vertical line $2$~time steps after the start of that signal).
            \item At the same time as the signal of speed~$-\frac 13$ is created, a signal of speed~$-\frac 12$ is 
                generated. When this signal collides with the bottommost signal, it bounces into
                a signal of speed~$\frac 14$ that will create the first SE-NW diagonal of the first square of
                the row of squares of the left, $4$~time steps after it will collide with the vertical line.
        \end{itemize}
\end{itemize}

On top of the construction states, except on the vertical lines, we add a parity layer $\{0, 1\}$: on a 
configuration, two neighboring cells of the construction must have different parity bits, 
otherwise a killer state appears. On the left of a vertical line there has to be parity~$1$ and
on the right parity~$0$, otherwise the killer state pops up again. This is to ensure that the columns 
will always contain an even number of squares.

\begin{figure}[ht!]
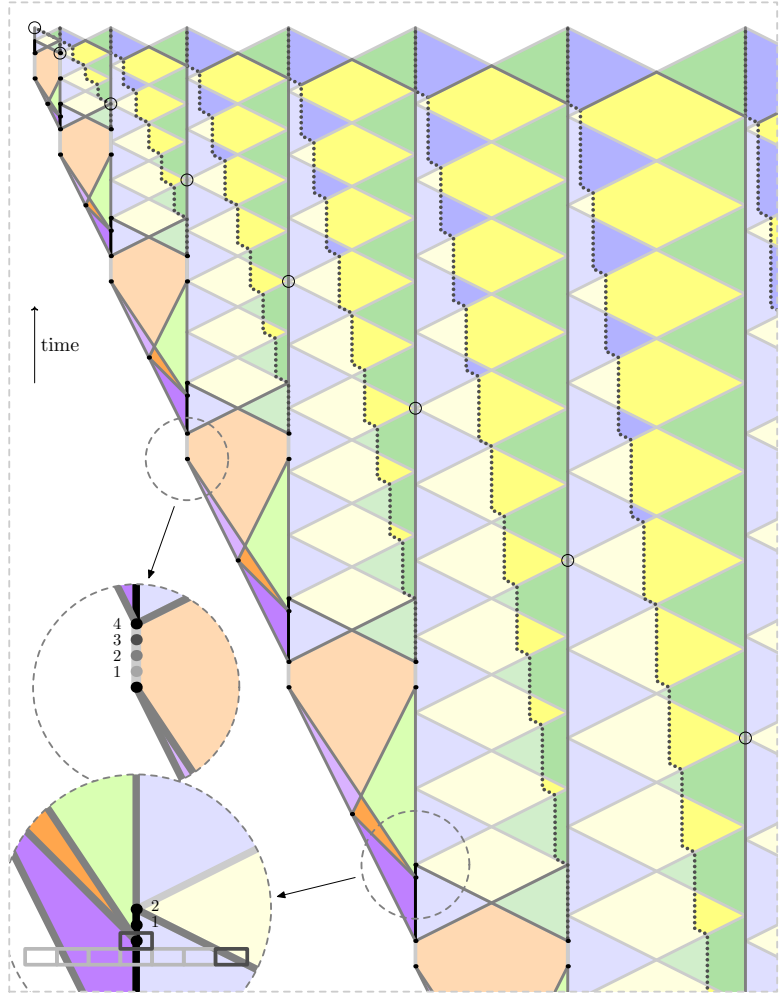

    \centering
    \scalebox{.75}{\includepicture{butterfly-0.mps}}
    \caption{\label{butterfly:baselayer} The sparse grid construction: it is based on columns containing
    a finite number of squares, whose number decreases when we go left. Note that the figure is 
    squeezed vertically.}
\end{figure}
The following lemmas address which types of configurations may occur in the limit set of this CA. First
note that any configuration wherein the construction states do not appear in the right order do not have a preimage.

\begin{lemma}\label{limitlem:squares}
    The sequence of preimages of a segment 
    ended by consecutive vertical lines 
    (and containing none) is a slice of a column of squares of even side.
\end{lemma}
\begin{proof}
    Suppose a configuration contains two vertical-line symbols, then to be in the limit set, in between
    these two symbols there needs to be two diagonal symbols, one for the SE-NW one and one for SW-NE one, 
    a symbol for the counting signal, and in between these signals there needs to be the appropriate 
    colors: there is only one possibility for each of them. If this is not the case, then the configuration
    has no preimage. 
    
    Also, the distance between the first vertical line and the SE-NW diagonal needs to be the
    same than the distance between the second vertical line and the SW-NE diagonal, otherwise the signals
    at the bottom --~the ones starting a column, that are the only preimages of the first diagonals~-- would have, in one case, created a vertical line in between, and in the other case, not started at the same time on 
    the right vertical.

    The side of the squares is even, otherwise the parity layer has no preimage.
\end{proof}
\begin{lemma}\label{limitlem:distances}
    A configuration of the limit set containing at least three vertical-line symbols needs to
    verify, for any three consecutive symbols, that if the distance between the first one and the second one is~$k$, then the distance between the second one and the third one needs to be $(k + 2)$.
\end{lemma}
\begin{proof}
    Let us take a configuration containing at least three vertical-line symbols, take three consecutive ones.
    The states between them have to be of the right form as we said above. Suppose the first of these 
    symbols is at distance~$k_1$ of the second one, which is at distance~$k_2$ of the third one. This means 
    that the first (resp. second) segment defines a column of squares of side~$k_1$ 
    (resp.~$k_2$). It is clear that the second column of squares cannot end before the first one.
    
    Now let~$i$ be the position of the counting signal of the first column and~$j$ the
    distance between the SW-NE diagonal and the left vertical line. The preimage of the first segment
    ends $(k_1i + j)$ (resp. $(k_1(i - 1) + j)$) steps before if the counting signal is on the left (resp. right)
    of the SW-NE diagonal. Then, the preimages of the left and right vertical lines of this column are the
    creating signals. Before the signal created on the right bounces on the one of speed~$-\frac 14$ created on the
    left, it collides with the one of speed~$-\frac 13$, thus determining the height of the squares on the right 
    column of squares. So $k_1 = k_2 - 2$.
\end{proof}

\begin{lemma}\label{limitlem}
    A configuration having two vertical-line symbols pertaining to the limit set needs to 
    verify one of the following statements.
    \begin{itemize}
        \item It is constituted of a finite number of vertical lines.
        \item It appears in the space-time diagram of fig.~\ref{butterfly:baselayer}.
        \item It is constituted of an infinite number of vertical lines, then starting
            from some position it is equal on the right to some (shifted) line of 
            fig.~\ref{butterfly:baselayer}.
    \end{itemize}
\end{lemma}
\begin{proof}
    We place ourselves in the case of a configuration of the limit set.
    Because of lemma~\ref{limitlem:squares}, two consecutive vertical lines at distance~$k$ from each other
    define a column of squares. In a space-time diagram they belong to, on their left there
    necessarily is another column of squares, because of the starting signal
    generated at the beginning of the left vertical line, except when $k = 3$, in which case there is
    nothing on the left. In this column, the vertical lines are at distance~$(k - 2)$, 
    see lemma~\ref{limitlem:distances}. So, if there is an infinite number of vertical lines,
    either it is of the form of fig.~\ref{butterfly:baselayer}, or there is some
    killer state coming from infinity on the left and ``eating'' the construction.
\end{proof}

\subsection{\label{computationingrid}Backward computation inside the grid}
We now wish to embed the computation of a reversible Turing machine inside the
aforementioned sparse grid, which for this purpose is better seen as a lattice. The fact
the TM is reversible allows us to embed it backwards in the CA. We will below denote by
\emph{TM time} (resp. \emph{CA time}) the time going forward for the Turing machine (resp.
the CA); on a space-time diagram, TM time goes from top to bottom, while CA time goes from
bottom to top (\cf arrows in fig.~\ref{computation:inbutterfly}). That way, the beginning
of the computation of the TM will occur in the first (topmost) square of the first
(leftmost) column of squares.

We have to ensure that any computation of the TM is possible, and in particular ensure that
such a computation is consistent over time; the idea is that at the first TM time step,
\ie the moment the sparse grid disappears, the tape is on each of the vertical line
symbols, but since these all disappear a finite number of CA steps before, we have to
compel all tape cells to shift to the right regularly as TM time increases.

Moreover, we want to force the presence of exactly one head (there could be none if
it were, for instance, infinitely far right). To do that, the grid is divided into three
parts that must appear in this order (from left to right): the left of the head, the right
of the head (together referred to as the computation zone) and the unreachable zone (where
no computation can ever be performed), resp. in blue, yellow and green in
fig.~\ref{computation:inbutterfly}.

The vertices of our lattice are the top left corners of the squares, each one marked by the
rebound of a SE-NW diagonal on a vertical line, while the top right corners will just serve
as intermediate points for signals. More precisely, if we choose (arbitrarily) the top left
corner of the first square of the first column to appear at site~$(0, 0)$, then for any $i,
j\in\NN$, the respective sites for the top left and top right corners of $s_{i, j}$, the
$(j + 1)$-th square of the $(i + 1)$-th column, are the following
(cf.~fig.~\ref{computation:inbutterfly}): 
\[
  \left\{\hspace{-1mm}
    \begin{array}{
      l@{\;}l@{\;}l} s^\ell_{i, j} &=& (i(i + 1), -2(i + 1)j)\\
      s^r_{i, j} &=& ((i + 1)(i + 2), -2(i + 1)j)\text{.} 
    \end{array}\right.  
\]

Fig.~\ref{computation:mt} illustrates a computation by the TM, with the three
aforementioned zones, as it would be embedded the usual way (but with reverse time) into a
CA, with site~$(i, -t)$ corresponding to the content of the tape at $i\in\NN$ and TM
time~$t\in\NN$. 

Fig.~\ref{computation:inter} represents another, still simple, embedding, which is a
distortion of the previous one: the head moves every even time step within a tape 
that is shifted every odd time
steps, so that instead of site~$(i, -t)$, we have two sites, $(i + t, -2t)$ and $(i + t,
-2t - 1)$, resp. the \emph{computation site} (big circle on fig.~\ref{computation:inter})
and the \emph{shifting site} (small circle on fig.~\ref{computation:inter}). The head only
reads the content of the tape when it lies on a computation site. This type of embedding
can easily be realized forwards or backwards (provided the TM is reversible).

Our embedding, derived from the latter, is drawn on fig.~\ref{computation:inbutterfly}. The
``only'' difference is the replacement of sites $(i + t, -2t)$ and $(i + t, -2t - 1)$ by
sites $s^\ell_{i, t}$ and $s^\ell_{i, t + 1}$. Notice that as the number of squares in a
column is always finite, each square can ``know'' whether its top left corner is a
computation or a shifting site with a parity bit. More precisely, the $j$-th square (from
bottom to top) of a column has a computation site on its top left if and only if $j$ is
even.

Let $s_{i, j}$ be a square of our construction. $s^\ell_{i, j}$ is either a computation
site or a shifting site. In the latter case, it is supposed to receive the content of a
cell of the TM tape with an incoming signal of speed~$-1$. All it has to do is to send it
to $s^\ell_{i, j - 1}$ (at speed~$0$), which is a computation site. In the former case,
however, things a slightly more complicated. The content of the tape has to be transmitted
to $s^\ell_{i - 1, j - 1}$ (which is a shifting site). To do that, a signal of speed~$0$ is
sent and waits for site~$s^r_{i - 1, j}$, which sends the content to $s^\ell_{i - 1, j -
1}$ with a signal of speed~$-1$ along the SE-NW diagonal. The problem is to recognize which
$s^r$~site is the correct one. Fortunately, there are only two possibilities: it is either
the first or the second $s^r$~site to appear after (in CA time, of course) $s^\ell_{i, j}$
on the vertical line. The first case corresponds exactly to the unreachable zone (where
$j\leq i$), hence the result if the three zones are marked. The lack of other cases is due
to the number of $s_i$~squares, which is only $2(i + 1)$.

Another issue is the superposition of such signals. Here again, there are only two cases:
in the unreachable zone there is none, whereas in the computation zone a signal of
speed~$0$ from a computation site can be superimposed to the signal of speed~$0$ sent by
the shifting site just above it. As aforesaid, there is no other case because of the
limited number of $s_i$ squares. Thus, there is no problem to keep the number of states of
the CA finite, since the number of signals going through a same cell is limited to two at
the same time.

While the two parts of the computation zones are to be separated by the presence of a head,
the unreachable zone is at the right of a signal that is sent from any computation site that
has two diagonals (one from the left and one from the right) below it (indicated as circles
on fig.~\ref{butterfly:baselayer}), goes at speed~$0$ until the next $s^r$~site, then at
speed~$1$ (along SE-NW diagonals) to the second next shifting site, and finally at
speed~$0$ again, to the next computation site (cf.~fig.~\ref{computation:inbutterfly}),
which also has two diagonals below it if the grid contains no error. Another way to detect
the unreachable zone is to detect that the counting signal crossed the SW-NE diagonal
exactly two CA time steps after it has crossed the SE-NW diagonal. This means that the 
unreachable zone is structurally coded in the construction.

Now only the movements of the head remain to be described (in black on
fig.~\ref{computation:inbutterfly}). Let $s^\ell_{i, j}$ be a computation site containing
the head.  
\begin{itemize} 
  \item If the previous move of the head (previous because we are
      in CA time, that is, in reverse TM time) was to the left, the next computation site
      is the one just above, that is, $s^\ell_{i, j - 2}$. The head is thus transferred by a
      simple signal of speed~$0$.  
  \item If the previous move was to stand still, the next
        computation site is $s^\ell_{i - 1, j - 2}$. It can be reached by a signal of
        speed~$0$ until the second next $s^r$~site, from which a signal of speed~$-1$
        (along a SE-NW diagonal) is launched, to be replaced by another signal of speed~$0$
        from $s^\ell_{i - 1, j - 1}$ on.  
  \item If the previous move was to the right, the
          next computation site is $s^\ell_{i - 2, j - 2}$. It can be reached by a signal
          of speed~$0$ until the second next $s^r$~site, from which a signal of speed~$-1$
          (along a SE-NW diagonal) is launched, to be replaced by another signal of
          speed~$0$ from $s^\ell_{i - 1, j - 1}$ on, which itself waits for the next
          $s^r$~site (which is $s^r_{i - 2, j}$) to start another signal of speed~$1$
          (along a SW-NE diagonal) that is finally succeeded to by a last signal of
          speed~$0$ from $s^\ell_{i - 2, j - 1}$ on.  
\end{itemize}

\begin{figure}[htp!]
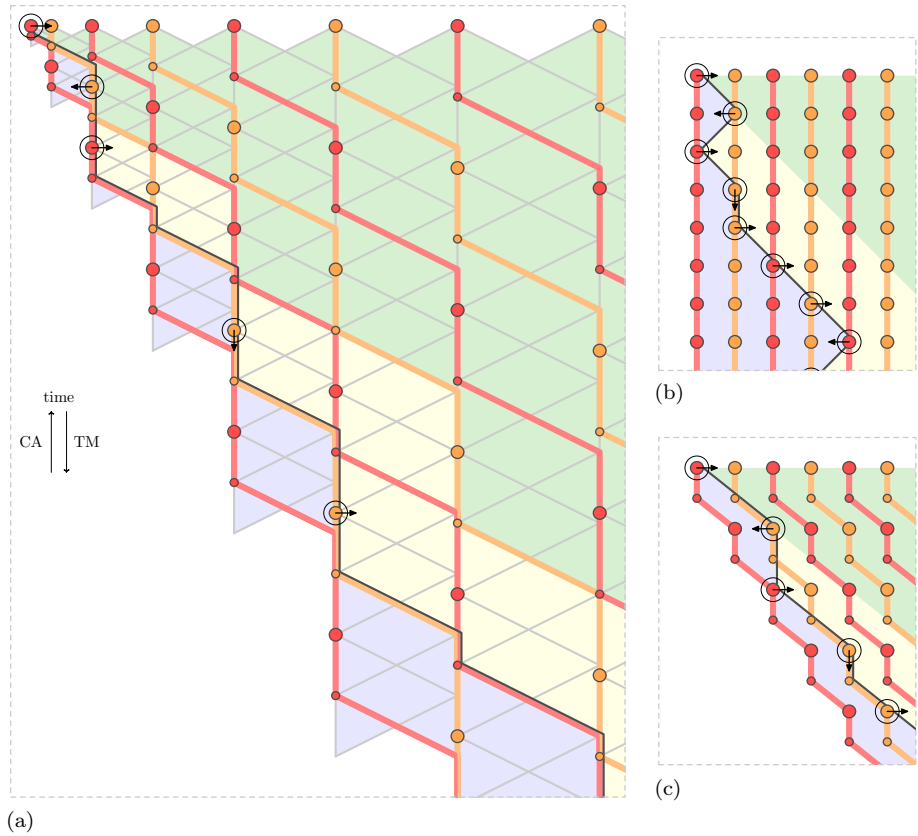
 \centering \begin{minipage}{.66\linewidth}
\hspace{-.25cm}\subfloat[]{\label{computation:inbutterfly}\scalebox{.6}{\includepicture{butterfly-1.mps}}}
\end{minipage} \hspace{.25cm}\begin{minipage}{.3\linewidth}
  \subfloat[]{\label{computation:mt}\scalebox{.6}{\includepicture{butterfly-2.mps}}}\\
\subfloat[]{\label{computation:inter}\scalebox{.6}{\includepicture{butterfly-3.mps}}}
\end{minipage} \caption{\label{butterfly:computation}The embedding of a Turing machine
computation in the sparse grid (\ref{computation:inbutterfly}), compared to the usual
embedding (\ref{computation:mt}) and a slightly distorted one (\ref{computation:inter}).
The paths followed by the content of each cell of the tape are in red and orange (two
colors just to keep track of the signals), while the one of the head is in black. The
arrows indicate the next move of the head (for TM time, going towards the bottom). The
green background denotes the zone the head cannot reach, while the computation zone is in
blue on the left of the head and in yellow on its right.} \end{figure}

\subsection{\label{hooper}The computation itself}
As we said before, the computation will take place on the computation sites, which will
contain two kinds of tape cells: one for the oracle and one for the work. In the unreachable
zone there are only oracle cells, which do not change over time except for the shifting.
Now we want to eliminate all space-time diagrams corresponding to rejecting computations 
of some Turing machine $M$. \citet{Ben1973} has proved that for any Turing machine, we
can construct a reversible one computing the same function. So a first idea would just 
be to encode this reversible Turing machine in the sparse grid; however there is no way
to guarantee that the work tape that was non-deterministically inherited from the
past corresponds to a valid configuration and by the time the Turing machine ``realizes'' 
this it will be too late, there will already exist configurations containing some oracle
that we would otherwise have rejected.

The solution to this problem is to use a robust Turing machine in the sense of \citet{Hoo1966},
that is to say a Turing machine that regularly rechecks its whole computation. \citet{KO2008}
have constructed reversible such machines. In these constructions the machines 
constructed were working on a bi-infinite tape, which had the drawback that some infinite
side of the tape might not be checked; here it is not the case, hence we can modify the machine
so that on an infinite computation it visits all cells of the tape (we omit the details for
brevity's sake).

In terms of limit sets, this means that if some oracle is rejected by the machine, then it
must have been rejected an infinite number of times in the past (CA time). So, only oracles 
pertaining to the desired class may appear in the limit set. 

Furthermore, even if some killer state coming from the right eats the grid, at some point in 
the past of the CA, it will be in the unreachable zone, and stay there for ever, so the 
computation from that moment on even ensures that the oracle computed is correct. Though,
that doesn't matter, because in this case the configurations of the corresponding space-time
diagram that are in the limit set are uniform both on the right and on the left except for 
a finite part in the middle, and are hence computable.

\section{\label{CB}Cantor-Bendixson rank of limit sets}

The \emph{Cantor-Bendixson derivative} of some set $S\subseteq \Sigma^\ZZ$, with $\Sigma$ finite,
is noted $\CBd{S}$ and consists of all configurations of~$S$ except the isolated ones. A 
configuration $\cacf c$ is said to be \emph{isolated} if there exists a pattern~$P$ such that 
$\cacf c$ is the only configuration of~$S$ containing $P$ (up to a shift). For any
ordinal $\lambda$ we can define $\CBdn{S}{\lambda}$, the Cantor-Bendixson derivative of 
rank $\lambda$, inductively:
\[
	\begin{array}{l@{\;\;}c@{\;\;\;}l}
		\CBdn{S}{0} &=& S\\
		\CBdn{S}{\lambda + 1} &=& \CBd{\CBdn{S}{\lambda}}\\
		\CBdn{S}{\lambda} &=& \displaystyle{\bigcap_{\gamma<\lambda}}\CBdn{S}{\gamma}\text{.}
	\end{array}
\]

The \emph{Cantor-Bendixson rank} of $S$, denoted by $\CB{S}$, is defined as the first ordinal ~$\lambda$ such that $\CBdn{S}{\lambda + 1} = \CBdn{S}{\lambda}$. In particular, when $S$ is countable,
$\CBdn{S}{\CB{S}}$ is empty. An element~$s$ is of rank~$\lambda$ in~$S$ if $\lambda$ is the
least ordinal such that $s\notin\CBdn{S}{\lambda}$. For more information about Cantor-Bendixson
rank, one may skim~\cite{Kechris}.

The Cantor-Bendixson rank corresponds to 
the height of a configuration corresponding to a preorder on patterns as noted by \citet{BDJ2008}{}.
Thus, it gives some information on the way the limit set is structured pattern-wise. A 
straightforward corollary of the construction above is the following.

\begin{corollary}\label{CBrank}
    There exists a constant $c\leq 10$ such that for any \pizu class $S$, there exists a
    CA $\ca A$ such that 
    \[ 
        \CB{\limitset{\ca A}} = \CB{S}+c\text{.}
    \]
\end{corollary}

Here the constant corresponds to the pattern overhead brought by the sparse-grid construction.

\section*{Acknowledgments}\label{sec:Acknowledgments}

This work was sponsored by grants EQINOCS ANR 11 BS02 004 03 and TARMAC ANR 12 BS02 007 01.
The authors would like to thank Nicolas Ollinger and Bastien Le Gloannec for some useful discussions.

\printbibliography

\end{document}